\documentclass[12pt]{article}

\usepackage{graphicx}
\usepackage{amssymb, amsmath, dsfont, amsthm, mathtools}
\usepackage{hyperref}
\usepackage{tikz}
\usepackage{multirow}
\usepackage{caption}
\usepackage{subcaption}
\usepackage{float}
\usepackage[margin = 2.5cm]{geometry}
\usepackage{authblk}

\usetikzlibrary{shapes,decorations,arrows,calc,arrows.meta,fit,positioning,shadings}
\tikzset{
    -Latex,auto,node distance =1 cm and 1 cm,semithick,
    state/.style ={ellipse, draw, minimum width = 0.7 cm},
    point/.style = {circle, draw, inner sep=0.06cm,fill,node contents={}},
    bidirected/.style={Latex-Latex,dashed},
    el/.style = {inner sep=2pt, align=left, sloped}
}
\newtheorem{theorem}{Theorem}
\newtheorem{corollary}{Corollary}

\usepackage{setspace}

\usepackage[authoryear,round]{natbib}

\title{Nonparametric Bounds for Evaluating the Clinical Utility of Treatment Rules}

\author[1]{Johannes Hruza \thanks{corresponding authors: $\{$johannes.hruza, samir.bhatt$\}$@sund.ku.dk}}
\author[2,3]{Erin Gabriel}
\author[4]{Arvid Sjölander}
\author[1,5]{Samir Bhatt$^*$}
\author[2,3]{Michael Sachs}

\affil[1]{Section of Health Data Science and AI, University of Copenhagen, Copenhagen, Denmark}
\affil[2]{Section of Biostatistics, University of Copenhagen,  Copenhagen, Denmark}
\affil[3]{Pioneer Centre for SMARTbiomed, University of Copenhagen, Copenhagen, Denmark}
\affil[4]{Department of Medical Epidemiology and Biostatistics, Karolinska Institutet, Stockholm, Sweden}
\affil[5]{MRC Centre for Global Infectious Disease Analysis, Imperial College London, London, UK}

\date{\today}

\begin{document}

\maketitle
Evaluating the value of new clinical treatment rules based on patient characteristics is important but often complicated by hidden confounding factors in observational studies. Standard methods for estimating the average patient outcome if a new rule were universally adopted typically rely on strong, untestable assumptions about these hidden factors. This paper tackles this challenge by developing nonparametric bounds—a range of plausible values—for the expected outcome under a new rule, even with unobserved confounders present. We propose and investigate two main strategies for derivation of these bounds. We extend these techniques to incorporate Instrumental Variables (IVs), which can help narrow the bounds, and to directly estimate bounds on the difference in expected outcomes between the new rule and an existing clinical guideline. In simulation studies we compare the performance and width of bounds generated by the reduction and conditioning strategies in different scenarios. The methods are illustrated with a real-data example about prevention of peanut allergy in children. Our bounding frameworks provide robust tools for assessing the potential impact of new clinical treatment rules when unmeasured confounding is a concern.  

\section{Introduction}

While numerous risk scores are available in clinical practice, many have not been effectively evaluated for their role in guiding medical decision-making. One formal approach to evaluating such scores is by assessing the expected outcome if treatment decisions were based on the score, often framed within the context of dynamic treatment regimes (DTRs) or individualized treatment rules (ITRs). A DTR (or ITR, for a single treatment point) is a sequence of treatment rules, mapping patient covariates to recommended treatments \citep{Murphy2003, Chakraborty2013, Tsiatis2019}. The value of a specific rule, $f$, derived from a risk score based on covariates $\boldsymbol{X}$, can be quantified by the average outcome observed if all subjects in the population of interest were treated according to that rule. Existing methods for estimating and optimizing this quantity often rely on the strong assumption of no unmeasured confounding given the observed covariates \citep{Zhao2012}.

However, in many real-world observational settings, particularly when evaluating a new rule $f(\boldsymbol{X})$, it is plausible that unmeasured  factors $\boldsymbol{U}$ influence both the treatment $A$ actually received and the outcome $Y$. These unmeasured factors (for example visual appearance to a clinician) might already be implicitly considered in the current Standard of Care (SOC), leading to confounding bias in later analysis if ignored.
Assessing the clinical utility \citep{Sachs2020, Hruza2025, Brand23} of a rule $f(\boldsymbol{X})$, defined as the improvement over the current standard of care, $CU(f) = E[Y(A=f(\boldsymbol{X}))] - E[Y]$, ideally accounts for such complexities. 
Here, $f(X)$ varies across subjects, and $Y(A=f(X))$ denotes the outcome had each subject’s treatment $A$ been set to their specific rule-based value.
While the SOC $E[Y]$ is typically observable, estimating the counterfactual outcome under the treatment rule $E[Y(A=f(\boldsymbol{X}))]$ remains challenging under unmeasured confounding.

Partial identification provides a framework for deriving valid bounds on causal estimands without relying on the common assumption of no unmeasured confounding \citep{Manski1990, Manski2003}. Researchers have also investigated the use of mild, interpretable assumptions—often through sensitivity analysis—to further narrow these bounds   \citep{Yadlowsky2022, Dorn2021, Ding2016, Pena2017}. Several authors have applied bounding techniques, often using Instrumental Variables (IVs), to estimate treatment effects or optimize DTRs \citep{Balke1997, Zhang2012, Qiu2020, Cui2021}. For instance, \citet{Balke1997} pioneered the use of linear programming (LP) to derive sharp bounds for the Average Treatment Effect (ATE) in a binary IV setting. This LP approach has been generalized by \citet{Sachs2022} for a broader class of causal queries and DAGs. Recent work has also explored bounds for super-optimal treatment rules \citep{Laurendeau2024}.

Despite these advances, bounding the specific quantity $E[Y(A=f(\boldsymbol{X}))]$ relevant for clinical utility in the presence of unmeasured confounders and potentially an IV remains an area needing further exploration. Our goal is to establish nonparametric bounds on $E[Y(A=f(\boldsymbol{X}))]$ under such conditions, as directly applying LP methods to the problem is computationally infeasible and grows  complexity with the levels of $X$.  We investigate two primary strategies:
\begin{enumerate}
    \item[(i)]\textbf{Reduction:} ignoring the information in the covarites $\boldsymbol{X}$ given the information contained in $f(\boldsymbol{X})$, thereby simplifying the causal graph and applying existing bounding machinery.
    \item[(ii)] \textbf{Conditioning:} Explicitly conditioning on $\boldsymbol{X}$, deriving bounds within strata defined by $\boldsymbol{X}=\boldsymbol{x}$, and then marginalizing these bounds over the distribution of $\boldsymbol{X}$.
\end{enumerate}
We compare these approaches, analyze their theoretical justification, and explore their extension to settings with IVs and additional observed covariates $W$ as well under the influence of existing guidelines $g$ that also map the covariates $\boldsymbol{X}$ to a treatment $A$. We aim to understand conditions under which these methods yield valid and potentially sharp bounds for clinical utility assessment.

We structure the paper as follows: In Section \ref{section: setting} we introduce the notation and causal settings. In Section \ref{section: reduction} we describe the DAG reduction approach for bounding $E[Y(A=f(\boldsymbol{X}))]$. In Section \ref{section: conditioning} we describe the conditioning approach. In Section \ref{section: iv_and_beyond} we extend these methods to scenarios involving instrumental variables and additional covariates $W$ as well as guidelines $g$. In Section \ref{section: simulations} we present simulation results comparing the two approaches. In Section \ref{section: real data} we apply the techniques to real-world data in the field of immunology. Finally, in Section \ref{section: discussion} we discuss the findings, challenges, and implications.

\section{Setting and Notation} \label{section: setting}

We consider an observational setting where we are interested in the  impact of applying a treatment rule $f$ based on observed covariates $\boldsymbol{X}$. Let $A$ be the treatment variable, which we assume to be categorical with levels $\mathcal{A} = \{a_0, a_1, \dots, a_{k-1}\}$. Let $Y$ be the outcome variable, assumed binary ($Y \in \{0, 1\}$) for simplicity, representing an event like recovery or mortality. Let $\boldsymbol{X}$ be a vector of observed covariates, taking values in a set $\mathcal{X}$. Let $\boldsymbol{U}$ represent a set of unmeasured variables that confound the relationship between $A$ and $Y$. We make no distributional assumption on $\boldsymbol{U}$.

We are given a deterministic function $f: \mathcal{X} \to \mathcal{A}$, representing the treatment rule or risk score of interest. This rule maps the covariates $\boldsymbol{X}$ of an individual to a recommended treatment level $f(\boldsymbol{X}) \in \mathcal{A}$. Here, $f(\boldsymbol{X})$ is a random variable representing the treatment recommendation dictated by the rule $f$.

Our primary causal estimand is the expected counterfactual outcome under rule $f$:
\begin{align*}
    \theta_f = E[Y(A=f(\boldsymbol{X}))] = P(Y(A=f(\boldsymbol{X})) = 1),
\end{align*}
where $Y(a)$ denotes the potential outcome if treatment $a$ were assigned. The clinical utility can then be expressed as $CU(f) = \theta_f - E[Y]$. Initially we assume $E[Y]$ is point identifiable from the observed data distribution $P(Y)$, and our focus is on bounding $\theta_f$.

We can expand $\theta_f$ using the law of total probability over the recommendations $f(\boldsymbol{X})$:
\begin{align*}
    \theta_f = \sum_{a_j \in \mathcal{A}} P(Y(a_j)=1, f(\boldsymbol{X})=a_j).
\end{align*}

We will analyze this quantity and the clinical utility under different causal assumptions represented by Directed Acyclic Graphs (DAGs) and nonparametric structural equation models (NPSEM).

\subsection{Non-IV Setting}

Initially, we consider a setting without an IV. The causal relationships are depicted in Figure \ref{fig: full models} a). We assume the following NPSEM:
\begin{align*}
\boldsymbol{u} &\sim P(\boldsymbol{u}) \\
\boldsymbol{x} &= h_{X}(\boldsymbol{u},\epsilon_{X})\\
a &= h_A(\boldsymbol{u},\boldsymbol{x}, \epsilon_{A}) \\
y &= h_Y(a,\boldsymbol{u}, \boldsymbol{x},\epsilon_{Y}) \\
f(\boldsymbol{x}) &=f(\boldsymbol{x}) \quad \text{(deterministic in observed distribution)}
\end{align*}
where $h_X, h_A, h_Y$ are unknown functions and $\epsilon_X, \epsilon_A, \epsilon_Y$ are exogenous error terms independent of each other and of $\boldsymbol{U}$. The variable $f(\boldsymbol{X})$ is fully determined by $\boldsymbol{X}$.

A direct naïve approach to find sharp bounds is to rewrite the problem as
\begin{align}
\label{eq: naive}
    \theta_f = \sum_{\boldsymbol{x}_j \in \mathcal{X}} P(Y(f(\boldsymbol{x}_j))=1, \boldsymbol{X}=\boldsymbol{x}_j)
\end{align}
and then bound this quantity using the LP algorithm \verb|causaloptim| by \cite{Sachs2022}. There are several drawbacks with this approach. First, to satisfy the necessary assumptions of \verb|causaloptim| we have to assume that $\mathcal{X}$ is finite. The second is due to the algorithmic complexity which scales superexponentially with respect to the levels of $\mathcal{X}$. This makes it infeasible as most covariates (e.g. vital signs of a patients) are continuous and even if finite, it quickly becomes infeasible to derive the bounds in this way. Our goal is to find bounds for $\theta_f = E[Y(A=f(\boldsymbol{X}))]$ based on the observed data distribution $P(A, Y, \boldsymbol{X})$ for continuous and multivariate $\boldsymbol{X}$.

\begin{figure}
\captionsetup[subfigure]{font=footnotesize, position=bottom}
\caption{full causal models}
\label{fig: full models}
\centering
\subcaptionbox{Non IV }[.4\textwidth]{
\begin{tikzpicture}[node distance =0.85 cm and 0.85 cm]
    \node (a) [label = below left: $A$,  point];
    \node (y) [label = right: Y, point, right  = 2 of a];
    \node (u) [label = below: $\boldsymbol{U}$, point, below right = of a];
    \node (x) [label = above: $X$, point, above right = of a];
    \node (b) [label = above: $f(X)$, point, right = of x];

    \path (a) edge (y);
    \path[dotted] (u) edge (a);
    \path[dotted] (u) edge (y);
    \path[dotted] (u) edge (x);
    \path (x) edge (a);
    \path (x) edge (y);
    \path (x) edge (b);
 \end{tikzpicture}}
\subcaptionbox{Non IV with $g$}[.4\textwidth]{
\begin{tikzpicture}[node distance =0.85 cm and 0.85 cm]
    \node (a) [label = below left: $A$,  point];
    \node (g) [label = above: $g(X)$, point, above = of a];
    \node (x) [label = above: $X$, point, right = of g];
    \node (f) [label = above: $f(X)$, point, right = of x];
    \node (y) [label = right: Y, point, right  = 2 of a];
    \node (u) [label = below: $\boldsymbol{U}$, point, below right = of a];

    \path (a) edge (y);
    \path[dotted] (u) edge (a);
    \path[dotted] (u) edge (y);
    \path[dotted] (u) edge (x);
    \path (x) edge (a);
    \path (x) edge (y);
    \path (x) edge (b);
    \path (g) edge (a);
    \path (x) edge (g);
\end{tikzpicture}}

\subcaptionbox{IV}[.4\textwidth]{
\begin{tikzpicture}[node distance =0.85 cm and 0.85 cm]
    \node (a) [label = below left: $A$,  point];
    \node (y) [label = right: Y, point, right  = 2 of a];
    \node (u) [label = below: $\boldsymbol{U}$, point, below right = of a];
    \node (x) [label = above: $X$, point, above right = of a];
    \node (b) [label = above: $f(X)$, point, right = of x];
    \node (z) [label = left: $Z$, point, left = of a];

    \path (a) edge (y);
    \path[dotted] (u) edge (a);
    \path[dotted] (u) edge (y);
    \path[dotted] (u) edge (x);
    \path (x) edge (a);
    \path (x) edge (y);
    \path (x) edge (b);
    \path (z) edge (a);
 \end{tikzpicture}}
\subcaptionbox{IV with $g$}[.4\textwidth]{
\begin{tikzpicture}[node distance =0.85 cm and 0.85 cm]
    \node (a) [label = below left: $A$,  point];
    \node (y) [label = right: Y, point, right  = 2 of a];
    \node (u) [label = below: $\boldsymbol{U}$, point, below right = of a];
    \node (g) [label = above: $g(X)$, point, above = of a];
    \node (x) [label = above: $X$, point, right = of g];
    \node (b) [label = above: $f(X)$, point, right = of x];
    \node (z) [label = left: $Z$, point, left = of a];

    \path (a) edge (y);
    \path[dotted] (u) edge (a);
    \path[dotted] (u) edge (y);
    \path[dotted] (u) edge (x);
    \path (x) edge (a);
    \path (x) edge (y);
    \path (x) edge (b);
    \path (z) edge (a);
    \path (g) edge (a);
    \path (x) edge (g);
 \end{tikzpicture}}
\end{figure}
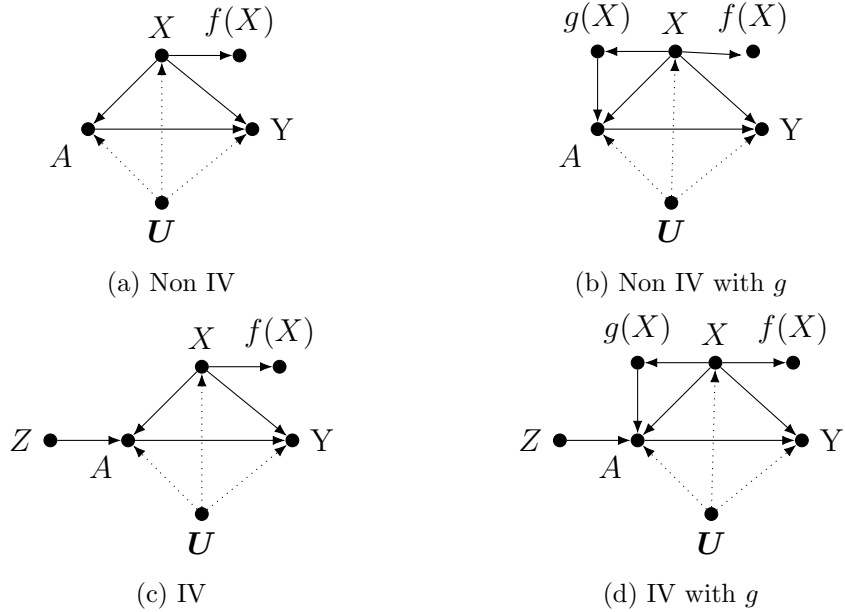

\section{Reduction}\label{section: reduction}
One strategy to obtain bounds on $\theta_f$ is to simplify the causal structure by ignoring the information in $\boldsymbol{X}$. This allows us to use bounding techniques developed for simpler graphical structures.
The validity of this approach relies on the principle that bounds derived for a model with fewer observed variables remain valid (though potentially not sharp) for a model where those variables are actually observed. We formalize this with the following Theorem.

\begin{theorem}[Valid Bounds Under Latent Variable Merging]
\label{theorem: merge latent}
    Let $G$ be a causal model with at least two latent variables $U_1, U_2$. Assume $U_1$ has no incoming edges, and $U_2$ has at most one incoming edge, which can only be from $U_1$ (i.e., $pa(U_1) = \emptyset$ and $pa(U_2) \subseteq \{U_1\}$). 
    
    Let $G'$ be the causal model constructed from $G$ by replacing $U_1$ and $U_2$ with a single latent variable $U = (U_1, U_2)$ that represents their joint distribution, with appropriate modifications to the structural equations.
    
    Let $\theta: \mathcal{P} \to \mathbb{R}$ be a causal parameter. Assume that the definition of $\theta$ does not involve interventions on
    $U_1$ or $U_2$. Then any valid bound for $\theta$ in $G'$ is also a valid bound for $\theta$ in $G$.
\end{theorem}

Theorem \ref{theorem: merge latent} states that if you have two ``exogenous" latent variables ($U_1$ being a source, $U_2$ only potentially depending on $U_1$), you can merge them into a single composite latent variable $U=(U_1,U_2)$ without losing the validity for bounding a causal parameter $\theta$ (as long as $\theta$ is not defined by intervening on $U_1$ or $U_2$ themselves). This is a simplification step: bounds derived in the simpler model $G'$ (with merged $U$) are still valid for the original, more complex model $G$.

Proof of Theorem is provided in the Supplementary materials.

\begin{figure}
\captionsetup[subfigure]{font=footnotesize, position=bottom}
\caption{Reduced causal models}
\label{fig: reduced models}
\centering
\subcaptionbox{Non IV without $g$}[.4\textwidth]{
\begin{tikzpicture}[node distance =0.85 cm and 0.85 cm]
    \node (a) [label = below left: $A$,  point];
    \node (y) [label = right: Y, point, right  = 2 of a];
    \node (u) [label = below: $(\boldsymbol{U} \text{,} \boldsymbol{X})$, point, below right = of a];
    \node (b) [label = above: $f(\boldsymbol{X})$, point, above right = of a];
    
    \path (a) edge (y);
    \path[dotted] (u) edge (a);
    \path[dotted] (u) edge (y);
    \path[dotted] (u) edge (b);
\end{tikzpicture}}
\subcaptionbox{Non IV with $g$}[.4\textwidth]{
\begin{tikzpicture}[node distance =0.85 cm and 0.85 cm]
    \node (a) [label = below left: $A$,  point];
    \node (y) [label = right: Y, point, right  = 2 of a];
    \node (u) [label = below: $(\boldsymbol{U} \text{,} \boldsymbol{X})$, point, below right = of a];
    \node (g) [label = above: $g(\boldsymbol{X})$, point, above right = 1 and 0.33 of a];
    \node (b) [label = above: $f(\boldsymbol{X})$, point, right = of g];

    \path (a) edge (y);
    \path (g) edge (a);
    \path[dotted] (u) edge (a);
    \path[dotted] (u) edge (y);
    \path[dotted] (u) edge (b);
    \path[dotted] (u) edge (g);
    
\end{tikzpicture}}

\subcaptionbox{IV without $g$}[.4\textwidth]{
\begin{tikzpicture}[node distance =0.85 cm and 0.85 cm]
    \node (a) [label = below left: $A$,  point];
    \node (y) [label = right: Y, point, right  = 2 of a];
    \node (u) [label = below: $\boldsymbol{(\boldsymbol{U} \text{,} \boldsymbol{X})}$, point, below right = of a];
    \node (b) [label = above: $f(\boldsymbol{X})$, point, above right = of a];
    \node (z) [label = left: $Z$, point, left = of a];

    \path (a) edge (y);
    \path[dotted] (u) edge (a);
    \path[dotted] (u) edge (y);
    \path[dotted] (u) edge (b);
    \path (z) edge (a);

\end{tikzpicture}}
\subcaptionbox{IV with $g$}[.4\textwidth]{
\begin{tikzpicture}[node distance =0.85 cm and 0.85 cm]
    \node (a) [label = below left: $A$,  point];
    \node (y) [label = right: Y, point, right  = 2 of a];
    \node (u) [label = below: $\boldsymbol{(\boldsymbol{U} \text{,} \boldsymbol{X})}$, point, below right = of a];
    \node (g) [label = above: $g(\boldsymbol{X})$, point, above right = 1 and 0.33 of a];
    \node (b) [label = above: $f(\boldsymbol{X})$, point, right = of b];
    \node (z) [label = left: $Z$, point, left = of a];

    \path (g) edge (a);
    \path (a) edge (y);
    \path[dotted] (u) edge (a);
    \path[dotted] (u) edge (g);
    \path[dotted] (u) edge (b);
    \path[dotted] (u) edge (y);
    \path (z) edge (a);

\end{tikzpicture}}

\end{figure}
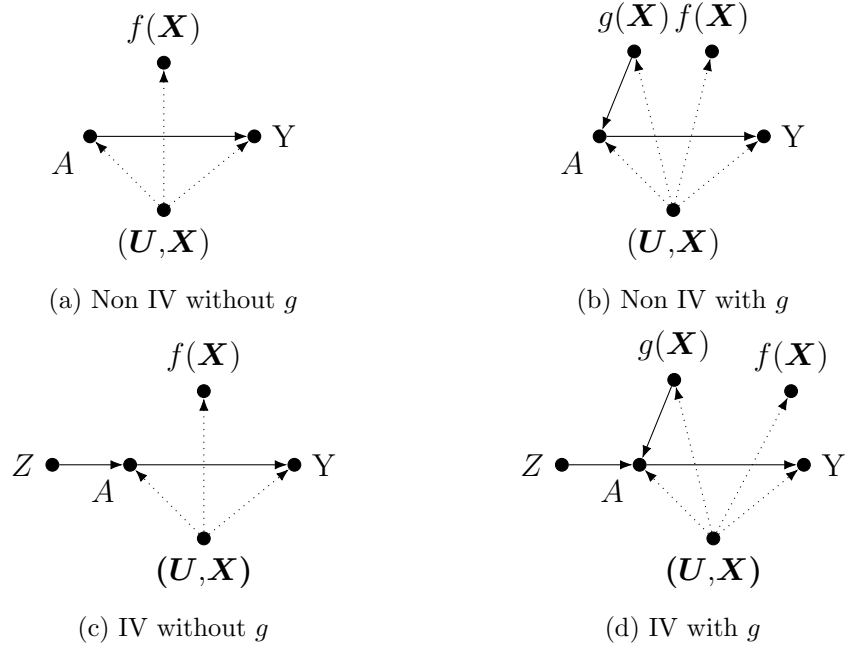

Consider the non-IV model in Figure \ref{fig: full models} a). We first treat $\boldsymbol{X}$ as unobserved or ignore the information contained in $\boldsymbol{X}$, while maintaining the information contained in $f(\boldsymbol{X})$. 
Then, applying Theorem \ref{theorem: merge latent}, we can merge the original latent variables $\boldsymbol{U}$ with the now-latent $\boldsymbol{X}$ into a single compound latent variable $\boldsymbol{U}^* = (\boldsymbol{U}, \boldsymbol{X})$. 
The modified NPSEM is then given by:
\begin{align*}
\boldsymbol{u^*} &\sim (P(\boldsymbol{u}), g_{X}(\boldsymbol{u},\epsilon_{X})) \\
a &= h_A(\pi_1(\boldsymbol{u^*}),\pi_2(\boldsymbol{u^*}), \epsilon_{A}) \\
y &= h_Y(a,\pi_1(\boldsymbol{u^*}), \pi_2(\boldsymbol{u^*}),\epsilon_{Y}) \\
f(\boldsymbol{x}) &= f(\pi_2(\boldsymbol{u^*})) \quad \text{(non deterministic in observed distribution)}, 
\end{align*}
where $\pi_1, \pi_2$ is the canonical projection from $\boldsymbol{U}^*$ to its first and second component. 
The resulting simplified causal structure is shown in Figure \ref{fig: reduced models} a).

In the reduced model of Figure \ref{fig: reduced models} a), the observed variables are $(A, Y, B)$. The variable $f(\boldsymbol{X})$ is now treated as an random variable. 

Using linear programming methods \citep{Balke1997, Sachs2022} applied to the structure in Figure \ref{fig: reduced models} a), we can derive sharp bounds for $\theta_f$ based on $P(A, Y, f(\boldsymbol{X}))$. For the case where $A$ is binary ($\mathcal{A}=\{a_0, a_1\}$) and $Y$ is binary, the sharp bounds on $\theta_f = P(Y(a_0)=1, f(\boldsymbol{X})=a_0) + P(Y(a_1)=1, f(\boldsymbol{X})=a_1)$ are given by:
\begin{align} \label{eq: binary non Iv}
 L_R &= p(A=0,Y=1,f\boldsymbol{X})=0)+p(A=1,Y=1,f\boldsymbol{X})=1) \\
 U_R &= 1 - P(A = 0, Y = 0, f\boldsymbol{X}) = 0) - P(A = 1, Y = 0, f\boldsymbol{X}) = 1)  \nonumber 
\end{align}

\begin{corollary}
The bounds $L_R$ and $U_R$ derived for the reduced DAG (Figure \ref{fig: reduced models} a)) are valid bounds for $\theta_f = E[Y(A=f(\boldsymbol{X}))]$ in the original non-IV DAG (Figure \ref{fig: full models} a) ).
\end{corollary}
\begin{proof}
from ignoring the information in $\boldsymbol{X}$ and Theorem \ref{theorem: merge latent} (merging $\boldsymbol{X}$ with $\boldsymbol{U}$).
\end{proof}

These bounds are guaranteed to be sharp for the simplified problem represented by Figure \ref{fig: reduced models} a). However, they may not be sharp for the original problem (Figure \ref{fig: full models} a)) as the reduction step discards all information in $P(\boldsymbol{X}|A,Y,B)$. Hence, the reduction step potentially discards information contained in the distribution $P(A, Y, \boldsymbol{X})$. Similar symbolic bounds can be computed using LP for cases where $A$ has more than two levels (e.g., for ternary $A$, we have written them out in the Supplement material).

\section{Bounding via Conditioning on Covariates} \label{section: conditioning}

Rather than treating $\boldsymbol{X}$ as latent, an alternative strategy uses the observed information by first computing bounds conditional on specific realizations $\boldsymbol{X}=\boldsymbol{x}$, followed by marginalization. Care is required when attempting to represent this conditioning within a causal model using a graphical structure. Specifically, one can mutilate the original DAG to reflect the conditioning and then apply causal bounds to this altered graph. However, this process is delicate, as conditioning in causal models (e.g., DAGs or NPSEMs) is not a closed operation: the post-conditioning structure may no longer be representable as a valid DAG or NPSEM. A well-known example is conditioning on a collider ($A \rightarrow B \leftarrow C$, conditioning on $B$), which induces dependencies that alter the graph's structure. Nevertheless, graphical criteria have been established by \cite{Richardson2002-bk} that identify when conditioning on a node yields a structure still representable as a DAG. These criteria are the foundations for our bounding method, which relies on such valid conditional structures. We provide a concise overview of these criteria and how to apply them in our setting in the Supplementary material.

\subsection{Conditional Bounds}

For a fixed value $\boldsymbol{X}=\boldsymbol{x}$, the treatment rule becomes deterministic: $f(\boldsymbol{x}) = a_{\boldsymbol{x}}$, where $a_{\boldsymbol{x}}$ is a specific treatment level in $\mathcal{A}$. We are interested in the conditional expectation $E[Y(A=f(\boldsymbol{X})) | \boldsymbol{X}=\boldsymbol{x}] = E[Y(A=a_{\boldsymbol{x}}) | \boldsymbol{X}=\boldsymbol{x}]$.

Conditioning on $\boldsymbol{X}=\boldsymbol{x}$ can simplify the causal structure among the remaining variables $(A, Y, \boldsymbol{U}, \dots)$. In the case of the causal model from Figure \ref{fig: full models} a), we apply the graphical criterion as described in the Supplementary material on how to condition on $\boldsymbol{X}= \boldsymbol{x}$. After the transformation the DAG is shown in Figure \ref{fig: conditional DAG} a). Alternatively we can look at the NPSEM which conditioned on $\boldsymbol{X}=\boldsymbol{x}$ is given by:

\begin{align*}
\boldsymbol{u} &\sim P(\boldsymbol{u} \mid x) \\
a &= g_A(\boldsymbol{u},\boldsymbol{x}, \epsilon_{A}) \\
y &= g_Y(a,\boldsymbol{u}, \boldsymbol{x},\epsilon_{Y}) \\
a_x &= f(\boldsymbol{x}) \quad \text{(deterministic in observed distribution)}
\end{align*}

Within the stratum $\boldsymbol{X}=\boldsymbol{x}$, we need bounds on $P(Y(a)=1 | \boldsymbol{X}=\boldsymbol{x})$ for any treatment level $a$ assigned by $f$. Let $L_a(\boldsymbol{x})$ and $U_a(\boldsymbol{x})$ denote sharp lower and upper bounds on $P(Y(a)=1 | \boldsymbol{X}=\boldsymbol{x})$, derived from the conditional observed data distribution and the appropriate conditional causal graph. These bounds are computed using LP methods (e.g. \verb|causaloptim| \citep{Sachs2022}) applied to the conditional problem. For example, in the simple confounded structure of Figure \ref{fig: conditional DAG} a), the symbolic bounds on $P(Y(a)=1 | \boldsymbol{X}=\boldsymbol{x})$ are functions of $P(A, Y | \boldsymbol{X}=\boldsymbol{x})$.

\begin{figure}
\captionsetup[subfigure]{font=footnotesize, position=bottom}
\centering
\caption{DAGs of causal models from figure \ref{fig: full models} after conditioning on $\boldsymbol{X}=\boldsymbol{x}$}
\label{fig: conditional DAG}
\subcaptionbox{Non IV }[.4\textwidth]{
\begin{tikzpicture}[node distance =0.85 cm and 0.85 cm]
    \node (a) [label = below left: $A$,  point];
    \node (y) [label = right: Y, point, right  = 2 of a];
    \node (u) [label = below: $\boldsymbol{U}$, point, below right = of a];
    \node (b) [label = above: $f(\boldsymbol{x})$, point, above = of y];

    \path (a) edge (y);
    \path[dotted] (u) edge (a);
    \path[dotted] (u) edge (y);
 \end{tikzpicture}}
\subcaptionbox{Non IV with $g$}[.4\textwidth]{
\begin{tikzpicture}[node distance =0.85 cm and 0.85 cm]
    \node (a) [label = below left: $A$,  point];
    \node (g) [label = above: $g(\boldsymbol{x})$, point, above = of a];
    \node (f) [label = above: $f(\boldsymbol{x})$, point, above = of y];
    \node (y) [label = right: Y, point, right  = 2 of a];
    \node (u) [label = below: $\boldsymbol{U}$, point, below right = of a];

    \path (a) edge (y);
    \path[dotted] (u) edge (a);
    \path[dotted] (u) edge (y);
    \path (g) edge (a);
\end{tikzpicture}}

\subcaptionbox{IV}[.4\textwidth]{
\begin{tikzpicture}[node distance =0.85 cm and 0.85 cm]
    \node (a) [label = below left: $A$,  point];
    \node (y) [label = right: Y, point, right  = 2 of a];
    \node (u) [label = below: $\boldsymbol{U}$, point, below right = of a];
    \node (b) [label = above: $f(\boldsymbol{x})$, point, above = of y];
    \node (z) [label = left: $Z$, point, left = of a];

    \path (a) edge (y);
    \path[dotted] (u) edge (a);
    \path[dotted] (u) edge (y);
    \path (z) edge (a);
 \end{tikzpicture}}
\subcaptionbox{IV with $g$}[.4\textwidth]{
\begin{tikzpicture}[node distance =0.85 cm and 0.85 cm]
    \node (a) [label = below left: $A$,  point];
    \node (y) [label = right: Y, point, right  = 2 of a];
    \node (u) [label = below: $\boldsymbol{U}$, point, below right = of a];
    \node (g) [label = above: $g(\boldsymbol{x})$, point, above = of a];
    \node (b) [label = above: $f(\boldsymbol{x})$, point, above = of y];
    \node (z) [label = left: $Z$, point, left = of a];

    \path (a) edge (y);
    \path[dotted] (u) edge (a);
    \path[dotted] (u) edge (y);
    \path (z) edge (a);
    \path (g) edge (a);
 \end{tikzpicture}}
\end{figure}
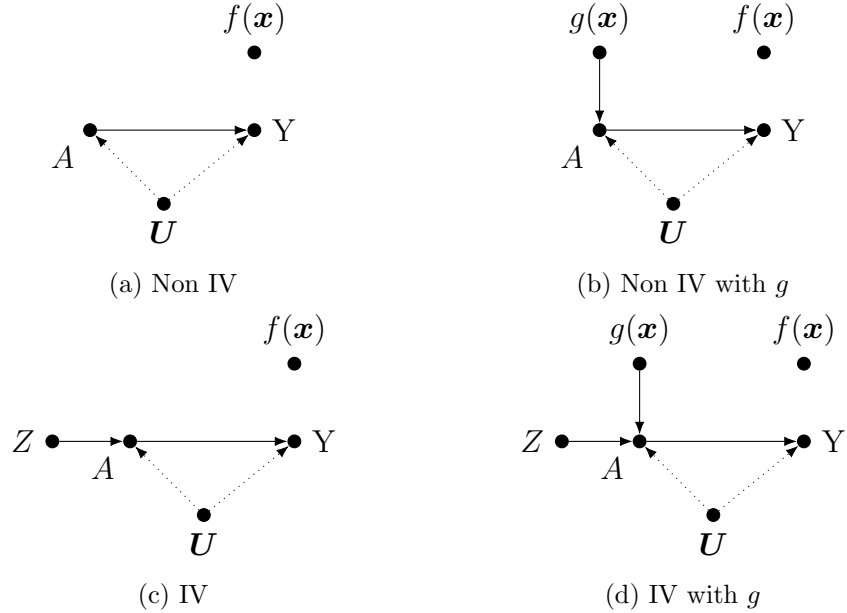

\subsection{Marginalization of Conditional Bounds}

To obtain bounds on the overall marginal quantity $\theta_f = E[Y(A=f(\boldsymbol{X}))]$, we use the law of total expectation and integrate the relevant conditional bounds over the distribution of $\boldsymbol{X}$:
\[ \theta_f = E_{\boldsymbol{X}} [ E[Y(A=f(\boldsymbol{X})) | \boldsymbol{X}] ] = E_{\boldsymbol{X}} [ P(Y(A=f(\boldsymbol{X}))=1 | \boldsymbol{X}) ] \]
Since for a given $\boldsymbol{x}$, the required potential outcome is $Y(A=f(\boldsymbol{x}))$, we integrate the bounds $L_{f(\boldsymbol{x})}(\boldsymbol{x})$ and $U_{f(\boldsymbol{x})}(\boldsymbol{x})$:

\begin{align}
     L_{C} &= \int_{\mathcal{X}} L_{f(\boldsymbol{x})}(\boldsymbol{x}) dP_{\boldsymbol{X}}(\boldsymbol{x}) =          \sum_{ a \in \{a_0, \dots, a_{k-1}\}} \int_{\{\boldsymbol{x} : f(\boldsymbol{x})=a\}} L_a(\boldsymbol{x}) dP_{\boldsymbol{X}}(\boldsymbol{x})
     \label{eq:cond_lower}\\
     U_{C} &= \int_{\mathcal{X}} U_{f(\boldsymbol{x})}(\boldsymbol{x}) dP_{\boldsymbol{X}}(\boldsymbol{x}) 
     =\sum_{a \in \{a_0, \dots, a_{k-1}\}} \int_{\{\boldsymbol{x} : f(\boldsymbol{x})=a\}} U_a(\boldsymbol{x}) dP_{\boldsymbol{X}}(\boldsymbol{x}),
     \label{eq:cond_upper}
\end{align}
where the last equality follows from splitting the integral over disjoint sets $\mathcal{X}=\bigsqcup_{a \in \mathcal{A}} \{\boldsymbol{x} \mid f(\boldsymbol{x})=a \}$.
\begin{corollary} \label{cor: valid bound conditional}
These marginalized bounds $[L_{C}, U_{C}]$ are valid for $\theta_f = E[Y(A=f(\boldsymbol{X}))]$    
\end{corollary}
\begin{proof}
    
This follows when looking at the RHS of \eqref{eq:cond_lower} and \eqref{eq:cond_upper}. For an arbitrary $a$ and since $\{L_a(\boldsymbol{x}) | a \in \mathcal{A}\}$ and $\{U_a(\boldsymbol{x}) | a \in \mathcal{A}\}$ are point wise valid their marginalized bounds over $\{\boldsymbol{x} \mid f(\boldsymbol{x})=a\}$ for each $a$ are also valid, see \citep[Proposition 1]{Gustav2025}. Last, we need to apply that summation of valid bounds are still valid bounds, as stated in \citep{Gabriel2022}. 
\end{proof}

When applying the estimands in Equations \ref{eq:cond_lower} and \ref{eq:cond_upper} to real-world settings, it is necessary to estimate the conditional probabilities involved. These estimates are then plugged into the symbolic bounds. Estimation approaches can range from simple models such as logistic regression to more complex methods including neural networks. Alternatively, influence function-based estimators, such as those proposed by \citet{Levis2025}, may also be employed.

While the conditional bounds $L_a(\boldsymbol{x}), U_a(\boldsymbol{x})$ might be sharp for the conditional probability $P(Y(a)=1 | \boldsymbol{X}=\boldsymbol{x})$, the marginalized bounds $L_{C}, U_{C}$ are not guaranteed to be sharp for the marginal estimand $\theta_f$. Sharpness can be lost during marginalization if the conditional bounds are not uniformly sharp \citep{Gustav2025}.

\section{IV and observed covariates} \label{section: iv_and_beyond}

The preceding discussion of bounds has not included IVs. Without IVs the causal null is unlikely to be excluded from the bounds. For the binary case we can see from Equation \ref{eq: binary non Iv} that the bounds have width one, implying that the causal null is always included.  Fortunately, the presented bounding techniques can be readily extended to incorporate IVs, additional observed covariates or guidelines; this extension is straightforward and relies entirely on the established theoretical results and tools.

\subsection{Instrumental Variable (IV) Setting}

Suppose we have an instrumental variable $Z$ that affects treatment $A$, is independent of the unmeasured confounders $\boldsymbol{U}$, and affects the outcome $Y$ only through $A$. The causal graph is shown in Figure \ref{fig: full models} c). The IV conditions provide additional constraints that can narrow the bounds on causal effects.

Both bounding strategies can be adapted:
\begin{enumerate}
    \item [(i)] Reduction with IV: Apply Theorem \ref{theorem: merge latent} to Figure \ref{fig: full models} c) (assuming the standard IV structure where $Z \perp \boldsymbol{U}$ and $\boldsymbol{U}$ does not cause $\boldsymbol{X}$). This yields the reduced DAG in Figure \ref{fig: reduced models} c). The distribution of the observed variables is now $P(A, Y, B \mid Z)$. The bounds $L_R, U_R$ for $\theta_f$ are computed using \verb|causaloptim| based on the observed distribution $P(Z, A, Y, B)$, exploiting the IV constraints. These bounds are generally narrower than the non-IV bounds. 
    
    \item [(ii)] Conditioning with IV: Condition on $\boldsymbol{X}=\boldsymbol{x}$ in Figure \ref{fig: full models} c). In the same way as before it can be shown that the resulting model is the DAG in \ref{fig: conditional DAG} c). The conditional DAG involves $(Z, A, Y, \boldsymbol{U})$. Within this stratum, $Z$ still acts as an instrument for the effect of $A$ on $Y$ (confounded by $\boldsymbol{U}$). We compute conditional bounds $L_a(\boldsymbol{x})$ and $U_a(\boldsymbol{x})$ for $P(Y(a)=1 | \boldsymbol{X}=\boldsymbol{x})$ using \verb|causaloptim|, based on the conditional distribution $P(Z, A, Y | \boldsymbol{X}=\boldsymbol{x})$. Then, we marginalize these bounds using Equations \ref{eq:cond_lower} and \ref{eq:cond_upper} to get $L_C, U_C$. Validity follows from Corollary \ref{cor: valid bound conditional}.
\end{enumerate}

\subsection{Observed Covariates W} \label{subsec:observed_w}

Consider the situation where, in addition to $\boldsymbol{X}$, we observe another set of covariates $\boldsymbol{W}$ which are additional covariates. These covariates $\boldsymbol{W}$ might influence the treatment $A$ and outcome $Y$ and could be related to the unobserved confounders $\boldsymbol{U}$, but they are not part of the treatment rule $f(\boldsymbol{X})$. Examples could include demographic variables or comorbidities not used in a specific risk score $f$. The causal structures are depicted in Figure \ref{fig: introducing W} for the non-IV and IV settings and the NPSEM are provided in the Supplementary material.
For the conditioning approach we can simply condition on $\boldsymbol{W}$ and $\boldsymbol{X}$ jointly, which still gives a valid DAG by \citet{Richardson2002-bk}. We then proceed as in section \ref{section: conditioning}.
For the reduction approach, we can adapt our bounding strategies to incorporate $\boldsymbol{W}$:

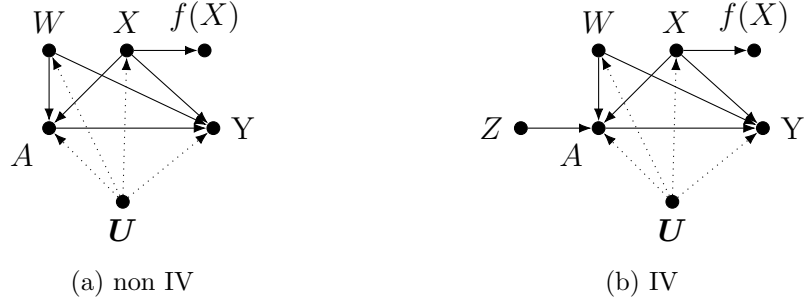
\begin{figure}
\captionsetup[subfigure]{font=footnotesize, position=bottom}
\caption{DAG with additional observed confounder $W$}
\label{fig: introducing W}
\centering
\subcaptionbox{non IV}[.4\textwidth]{
\begin{tikzpicture}[node distance =0.85 cm and 0.85 cm]
    \node (a) [label = below left: $A$,  point];
    \node (y) [label = right: Y, point, right  = 2 of a];
    \node (u) [label = below: $\boldsymbol{U}$, point, below right = of a];
    \node (w) [label = above: $W$, point, above = of a];
    \node (x) [label = above: $X$, point, right = of g];
    \node (b) [label = above: $f(X)$, point, right = of x];

    \path (a) edge (y);
    \path (w) edge (a);
    \path (w) edge (y);
    \path (x) edge (a);
    \path (x) edge (y);
    \path (x) edge (b);
    \path[dotted] (u) edge (a);
    \path[dotted] (u) edge (y);
    \path[dotted] (u) edge (x);
    \path[dotted] (u) edge (w);
\end{tikzpicture}}
\subcaptionbox{IV }[.4\textwidth]{
\begin{tikzpicture}[node distance =0.85 cm and 0.85 cm]
    \node (a) [label = below left: $A$,  point];
    \node (y) [label = right: Y, point, right  = 2 of a];
    \node (u) [label = below: $\boldsymbol{U}$, point, below right = of a];
    \node (w) [label = above: $W$, point, above = of a];
    \node (x) [label = above: $X$, point, right = of g];
    \node (b) [label = above: $f(X)$, point, right = of x];
    \node (z) [label = left: $Z$, point, left = of a];

    \path (a) edge (y);
    \path (w) edge (a);
    \path (w) edge (y);
    \path (x) edge (a);
    \path (x) edge (y);
    \path (x) edge (b);
    \path (z) edge (a);
    \path[dotted] (u) edge (a);
    \path[dotted] (u) edge (y);
    \path[dotted] (u) edge (x);
    \path[dotted] (u) edge (w);
    
\end{tikzpicture}}
\end{figure}

\paragraph{Reduction Approach with W:}
We follow the standard reduction procedure by treating $\boldsymbol{X}$ (the variables in the rule $f$) as unobserved and merging them with $\boldsymbol{U}$ into $\boldsymbol{U}^* = (\boldsymbol{U}, \boldsymbol{X})$ using Theorem \ref{theorem: merge latent}. But, we keep $\boldsymbol{W}$ as an observed variable. The reduced DAG now involves observed variables $(A, Y, B=f(X), \boldsymbol{W})$ (plus $Z$ in the IV case) and the latent $\boldsymbol{U}^*$. Bounds $L_R, U_R$ for $\theta_f$ are then computed using LP methods based on the observed joint distribution. The validity follows from Theorem \ref{theorem: merge latent}.

\section{Introducing Guidelines} \label{subsec:guidelines}

In \citet{Hruza2025}, a novel approach is taken by assessing the utility of a new rule $f(X)$ relative to an existing guideline $ g(X)$, offering better insights into decision-making under the given framework.
 The clinical utility is then defined as the difference in expected outcomes:
\begin{equation}
    CU(f, g) = E[Y(A=f(X))] - E[Y(A=g(X))] = \theta_f - \theta_g.
\end{equation}
We assume the guideline $g(X)$ may influence the observed treatment $A$, as depicted in the causal DAGs in Figures \ref{fig: full models} b) (non-IV) and d) (IV), where $G=g(X)$ is included as a node influencing $A$. Our aim is to bound $CU(f, g)$ directly using either the reduction or conditioning strategy, which may yield narrower bounds than bounding of $\theta_f$ and $\theta_g$ separately.

\subsubsection{Bounding the Clinical Utility Difference}

\paragraph{Reduction Approach:}
Following Theorem \ref{theorem: merge latent}, we simplify the causal structure by treating $X$ as unobserved, leading to the reduced DAGs in Figures \ref{fig: reduced models} b) (non-IV) and d) (IV). In these reduced models, the observed variables include $(A, Y, B=f(X), G=g(X))$ (and $Z$ in the IV case). The estimand $CU(f,g)$ can be expressed as the difference between sums of counterfactual probabilities involving $B$ and $G$:
\begin{align*} \label{eq:cu_estimand_reduced}
    CU(f,g) = \left( \sum_{a_j \in \mathcal{A}} P(Y(a_j)=1, B=a_j) \right) - \left( \sum_{a_k \in \mathcal{A}} P(Y(a_k)=1, G=a_k) \right).
\end{align*}
We apply LP optimization techniques (e.g., \verb|causaloptim|) directly to this linear objective function, subject to the constraints imposed by the observed joint distribution and the causal structure of the reduced DAG. This process yields sharp bounds relative to the reduced model, which are valid bounds for the original problem.

\paragraph{Conditioning Approach:}
Alternatively, we condition on $\boldsymbol{X}=\boldsymbol{x}$. This leads to the conditional DAGs in Figures \ref{fig: conditional DAG} b) (non-IV) and d) (IV). For each stratum $\boldsymbol{x}$, let $a_{\boldsymbol{x}} = f(\boldsymbol{x})$ and $a'_{\boldsymbol{x}} = g(\boldsymbol{x})$. We are interested in the conditional clinical utility difference:
\begin{align*}
        CU(\boldsymbol{x}) = P(Y(a_{\boldsymbol{x}})=1 | \boldsymbol{X}=\boldsymbol{x}) - P(Y(a'_{\boldsymbol{x}})=1 | \boldsymbol{X}=\boldsymbol{x}).
\end{align*}
Using LP methods on the conditional DAG and the conditional observed distribution ($P(A, Y | \boldsymbol{X}=\boldsymbol{x})$ or $P(Z, A, Y | \boldsymbol{X}=\boldsymbol{x})$), we compute sharp bounds $[L_{a_x,a_{x'}}(\boldsymbol{x}), U_{a_x,a_{x'}}(\boldsymbol{x})]$for this difference $CU(\boldsymbol{x})$ within the stratum.
The overall bounds for $CU(f,g)$ are then obtained by marginalizing these conditional bounds over the distribution of $\boldsymbol{X}$:
\begin{align*}
     L_{CU,C} &= \int_{\mathcal{X}} L_{a_x,a_{x'}}(\boldsymbol{x}) \, dP_{\boldsymbol{X}}(\boldsymbol{x}), \text{ and}\\
     U_{CU,C} &= \int_{\mathcal{X}} U_{a_x,a_{x'}}(\boldsymbol{x}) \, dP_{\boldsymbol{X}}(\boldsymbol{x}).
\end{align*}
As established previously, these marginalized bounds are valid (Corollary \ref{cor: valid bound conditional}) for the overall clinical utility difference $CU(f,g)$, although potentially not sharp.

Both strategies thus allow for bounding of the clinical utility difference relative to a guideline $g(X)$ under unmeasured confounding.

\section{Simulations}\label{section: simulations}

We have conducted simulations to showcase the validity of the bounds and empirically investigate which bounds are narrower for certain settings.
For the simulation we investigate the setting in figure: \ref{fig: full models} c), with the following domains: $Z,Y,U \in \{0,1\}$, $A \in \{0,1,2\}$ and $X \in \{0,1,2,3,4,5\}$. The dimensions are chosen such that it is unfeasible to apply LP directly, e.g. we have run a LP directly of the problem (see \eqref{eq: naive} on a singe core of an Apple M2 chip for more than 24 hours, after which we have stopped the LP. Hence we have to use our approaches. From this setup we sample a valid but random distribution. We do this by uniformly sampling a distribution for each vertex conditioned on it's parents.

Let $ v$ be a vertex with domain $ \mathcal{D}_v = \{v_1, \dots, v_d\}$ with cardinality $d$. To define the conditional probability distribution $P(v = v_i \mid \text{pa}(v))$, we proceed as follows.
First, we fix a particular assignment of values to the parent variables $ \text{pa}(v) $. For each such combination of parent values, we generate a separate conditional distribution over the values of $ v $.
Next, we generate $d$ independent random variables from the uniform distribution,  $r_1, r_2, \dots, r_d \sim \text{Uniform}(0,1)$:
Then, we normalize these values to obtain a valid probability distribution (e.g sums to $1$):
$$
P(v = v_i \mid \text{pa}(v)) = \frac{r_i}{\sum_{j=1}^{d} r_j}
$$
The above procedure is repeated for every possible assignment of the parent variables $\text{pa}(v)$, resulting in a complete conditional probability table for $v$.

Based on all conditional distribution we can compute the joint distribution, which we use to calculate the true value of counterfactual outcome of $E[Y(f(X))]$ as well as the bounds using the two proposed methods respectively. In this simulation we have chosen:
\begin{align*}
    f \colon \{0, \dots, 5\} &\to \{0, 1, 2\}\\
    \{0,1\} &\mapsto 0\\
    \{2,3\} &\mapsto 1\\
    \{4,5\} &\mapsto 2
\end{align*}

We have repeated these steps $10 \times 10^6$ times. 
We observe that in all simulations the true value of the target causal effect lies in between the bounds, e.g. the bounds are valid for both approaches. Interestingly, in all of the simulations the conditional bounds are either as narrow or narrower than the reduced bounds. For this setting this suggests that the information loss by ignoring X is much stronger than the possible widening of bounds by adding together two bounds conditional on X. A Figure demonstrating this can be found in the Supplementary material and the 
implementation details can be found here \url{https://github.com/jhruza/causal_bounds}.

We have performed additional simulations trying to find an example where the reduced bounds outperform the conditional. As we were unable to find such a case, we conjecture that the conditional bounds are always as narrow or narrower than the reduced bounds. This is not to say that we believe the conditional strategy provides sharp bounds given the information from the original, unconditional, DAG and NPSEM.

\section{Motivating data example}
\label{section: real data}

The Learning Early About Peanut Allergy study \citep{DuToit2015} investigated the impact of early peanut introduction on the development of peanut allergy in high-risk infants. The study enrolled 640 infants, aged between 4 and 11 months, who had severe eczema, egg allergy, or both. These participants were randomized into two groups: one group was assigned to consume $6$ g of peanut protein per week, while the other was assigned to avoid peanuts until they reached 60 months of age. The true amount of consumed peanut protein was also recorded. A key piece of information recorded for each participant prior to randomization was their initial sensitivity to peanuts, determined by a skin-prick test. The primary outcome measured was the presence of peanut allergy at 60 months, which was assessed through an oral food challenge. The randomization to either the peanut consumption or avoidance group serves as the instrumental variable in your analysis. The baseline skin-prick test result can be interpreted as a covariate for each child available in this dataset.
Clinical practice guidelines from the UK (1998) and the US (2000) recommended excluding allergenic foods from the diets of high-risk infants and their mothers during pregnancy and breastfeeding.
These guidlines were withdrawn in 2008 \citep{DuToit2015}.

Here the simple standard of care is not identified as $E\{Y\}$, as randomization changes the treatment given. We can however consider this a setting with guidelines, where we are only interested in the difference between enforcing those guidelines and the new treatment rule. This setting is as described in section \ref{subsec:guidelines} which corresponds to the DAG in Figure \ref{fig: full models} d).

In particular the IV $Z \in \{0,1\}$ is the randomization into the peanut avoidance or peanut consumption group.  The outcome $Y\in\{0,1\}$ is presence of peanut allergy at 60 month where $Y=1$ indicates a allergic reaction and $Y=0$ no allergic reaction; note that we have ignored the drop-out in the study, which was the point of the bounds analysis in \citet{Gabriel2021}. The covariate $X\in \{0,1\}$ is the result of baseline skin-prick test, where $X=1$ indicates a positive test.  The guideline $g\colon x \mapsto 0$ assigns everyone to avoid peanuts.
We investigate two ITRs. The first is simply the exposure all children to peanuts. The second is motivated by the observation that in clinical practice, individuals with mild allergic sensitivity are typically exposed to lower doses of the allergen, while those with more tolerance may receive higher doses to induce or maintain desensitization.:
\begin{enumerate}
    \item $f_1$ assigns everyone to peanut exposure. In this case the treatment $A \in \{0,1\}$ where $A=0$ is peanut exposure of $\leq$ 0.2 g per week and $A=1$ peanut exposure of $> 0.2$g per week. And $f_1=1$.
    \item $f_2$ also assigns everyone to peanut exposure but to different levels based on the skin-prick test. The treatment is given by $A\in \{0,1,2\}$ where $A=0$ is peanut exposure $\leq 0.2$ g per week, $A=1$ is $0.2 <$ grams of peanut exposure $\leq 6$ per week and $A=2$ is $6 <$ grams of peanut exposure per week. Formally $f(0):=1, f(1):=2$.
\end{enumerate}

Following Section \ref{subsec:guidelines}, the symbolic bounds are calculated using the reduction and the conditioning approach. The Clinical utility $CU(f, g) = E[Y(A=f(X))] - E[Y(A=g(X))] = P(Y(A=f(X))=1) - P(Y(A=g(X))=1)$ is the change in percentage the population would benefit from ITR $f$ with respect to $g$. That is a negative CU indicates $f$ is superior and a positive CU indicates the SOC $g$ is superior. 

For $f_1$, the bounds are given by $[-0.154, -0.109]$ using the reduction approach, and $[-0.153, -0.108]$ using the conditional approach. In both cases, zero is excluded from the interval indicating a consistent negative effect. If one is convinced that the causal mechanism is given by the DAG \ref{fig: full models} d), then that peanut exposure is clearly the preferable treatment. These results align with current understanding and are reflected in recent clinical guidelines \citep{Togias2017}. 

In this particular setting, we could also have applied the original Balke and Pearl bounds \cite{Balke1997}, since here the clinical utility coincides with the classical Average Treatment Effect (ATE) comparing treatment versus no treatment.

This, however, is not feasible when evaluating the clinical utility of $f_2 $. For $f_2$, the bounds are given by $[-0.164, 0.234]$ under the reduction approach and $[-0.163, 0.237]$ under the conditional approach. These bounds are notably wider and, although the lower bound is slightly more favorable (more negative), the range includes zero and positive values. This might seem counterintuitive at first, but it reflects the underlying uncertainty in the data.

It's important to avoid interpreting these bounds as confidence intervals as they serve a different purpose. While confidence intervals provide information about estimation precision and likelihood, bounds represent the full range of values consistent with the data under unmeasured confounding (which can be arbitrary high, but consistent with the assumed DAG). A higher upper bound simply indicates that such an outcome cannot be ruled out; it does not imply that it is probable.

\section{Challenges and Discussion}\label{section: discussion}
This paper addressed the problem of evaluating new clinical treatment rules $f(\boldsymbol{X})$ using observational data where unmeasured confounding may be present. We developed and investigated two main bounding strategies—Reduction and Conditioning—to estimate bounds for the expected outcome $E[Y(A=f(\boldsymbol{X}))]$ and the clinical utility $CU(f,g)$. These methods provide bounds that are derived without parametric assumptions, conditional on the correctness of the specified underlying causal DAG and NPSEM assumptions.

The reduction strategy, based on Theorem \ref{theorem: merge latent}, uses a simplified causal structure by treating $\boldsymbol{X}$ as unobserved. This enables the application of existing linear programming methods to a reduced problem. Its advantage is computational feasibility and nonparametric estimation of the obtained boumds. However, a disadvantage is potential information loss, as it marginalizes over $\boldsymbol{X}$ before bounding, which may lead to wider, yet still valid, bounds than if all information were used.

The conditioning strategy, in contrast, first derives bounds conditional on $\boldsymbol{X}=\boldsymbol{x}$ and then marginalizes these bounds over $P(\boldsymbol{X})$ (Equations \ref{eq:cond_lower} and \ref{eq:cond_upper}). This approach attempts to use all available information in $\boldsymbol{X}$. Our simulation studies (Section \ref{section: simulations}) suggested that, for the scenarios tested, Conditioning frequently produced narrower bounds and was never wider than the reduction strategy bounds. Additionally as we are stratifying over $\boldsymbol{X}$ it is computationally less intensive than the reduced bounds.

Future research directions include determining theoretical conditions for the sharpness of marginalized conditional bounds. Extending these frameworks to assist in the direct optimization of rules $f(\boldsymbol{X})$ under unmeasured confounding is also an area for further study.
Another consideration is that when using plug-in estimators, the sampling variability of the symbolic bounds increase as the number of levels or the dimension of $X$ grows.

In conclusion, the proposed DAG Reduction and Conditioning frameworks offer useful and transparent methods for assessing clinical treatment rules when unmeasured confounding is a concern, conditional on a correctly specified causal graph. They enable researchers to quantify uncertainty and improve judgments about the potential impact of new clinical guidelines.

\subsubsection*{Funding}
SB acknowledges funding from the MRC Centre for Global Infectious Disease Analysis (reference MR/X020258/1), funded by the UK Medical Research Council (MRC). This UK funded award is carried out in the frame of the Global Health EDCTP3 Joint Undertaking. SB is funded by the National Institute for Health and Care Research (NIHR) Health Protection Research Unit in Modelling and Health Economics, a partnership between UK Health Security Agency, Imperial College London and LSHTM (grant code NIHR200908). Disclaimer: “The views expressed are those of the author(s) and not necessarily those of the NIHR, UK Health Security Agency or the Department of Health and Social Care.”. SB acknowledges support from the Novo Nordisk Foundation via The Novo Nordisk Young Investigator Award (NNF20OC0059309) which also supports JH.  SB acknowledges the Danish National Research Foundation (DNRF160) through the chair grant.  SB acknowledges support from The Eric and Wendy Schmidt Fund For Strategic Innovation via the Schmidt Polymath Award (G-22-63345). EG and MS are funded in part by NNF Grant NNF22OC0076595.

\subsubsection*{Data Availability Statement}
The publicly available data for the LEAP trial were downloaded from the Immune Tolerance Network TrialShare website (\url{https://www.itntrialshare.org/}, study identifier: ITN032AD).

\bibliographystyle{biom}
\bibliography{references}
\end{document}